\documentclass[aps,pra,epsfig,showpacs,superscriptaddress,footnote]{revtex4}
\usepackage{amsmath,amsthm,amssymb}
\usepackage{graphicx}
\makeatletter
\def\be{\begin{equation}}
\def\ee{\end{equation}}
\def\bea{\begin{eqnarray}}
\def\eea{\end{eqnarray}}
\def\bma{\begin{mathletters}}
\def\ema{\end{mathletters}}
\def\bi{\begin{itemize}}
\def\ei{\end{itemize}}
\newtheorem{thm}{Theorem}

\newtheorem{lem}{Lemma}[thm]

\newcommand{\ket}[1]{ | \, #1  \rangle}

\begin{document}

\title{Local cloning of CAT states}
\author{Ramij Rahaman}
\email{Ramij.Rahaman@ii.uib.no} \affiliation{Department of Informatics, University of Bergen, Post Box-7803, 5020, Bergen,
Norway }

\begin{abstract}
In this Letter we analyze the (im)possibility of the exact cloning
of orthogonal three-qubit CAT states under local operation and
classical communication(LOCC) with the help of a restricted entangled
state. We also classify the three-qubit CAT
states that can(not) be cloned under LOCC restrictions and extend the results to the n-qubit case.

\end{abstract}


\maketitle

\section{Introduction}

Quantum cloning and quantum state distinguishability are two
important and closely related areas in quantum information theory. The
no cloning theorem \cite{wootters} says that only orthogonal
quantum states can be cloned perfectly in principle and only when the required operations are performed on the entire
system. But if we restrict ourselves to LOCC(Local Operation
and Classical Communication) then further restrictions on cloning
apply. For example, the very obvious first restriction will be that an
entangled blank state is needed to clone an entangled state. The
LOCC cloning of entangled states was first introduced by Ghosh
\emph{et al.} \cite{ghosh2004} where it was shown that any two orthogonal
two qubit maximally entangled states can be cloned under LOCC with the
help of one ebit of entanglement. However, there exist two
orthogonal two qubit non-maximal entangled states which can not be
cloned under LOCC with the help of one(or less) ebit of entanglement
\cite{kay,choudhary}. So the relation between amount of entanglement present in the unknown states and possibility of cloning of
those entangled states under LOCC is not simple as sometimes LOCC cloning of
two qubit less entangled states of a particular type is not possible with any two qubit
entangled state as a blank copy. Thus the study of LOCC cloning of orthogonal entangled states becomes
interesting when the ensemble states are not maximal. This study
is also important because LOCC cloning is very
closely connected with many important information processing
tasks, such as channel copying, quantum key distribution, error
correction, entanglement distillation and is also helpful in
understanding the nonlocality of a set of states\cite{owari,bennett}. Most work in this area is for
bipartite (maximally and nonmaximally) entangled states
\cite{ghosh2004,owari,owari1,anselmi,choudhary,kay,ramij,GYC10}.
Although some results are known for local copying of a set of
orthogonal three qubit maximally entangled states
\cite{choudhary1}, no result is known for non-maximally entangled states. We, in
this letter concentrate on LOCC cloning of orthogonal \emph{CAT states}\footnote{states of the form
  $\ket{\psi}=\cos{\alpha}\ket{\phi}+\sin{\alpha}\ket{\eta}$, where
$\langle \phi|\eta\rangle=0$ and $\alpha\in (0,\pi/4)$, \emph{i.e.}, an unequal superposition of two orthonormal basis states.}\cite{GKRSS00}. \\

We organize our paper as follows: in section II, we prove that no
set of orthogonal three qubit CAT states can be cloned by LOCC
with the help of any three qubit CAT state having the same amount
of entanglement in the blank copy. We also prove that the
cloning remains impossible for the set of states having a pair of kind (\ref{3ghzI}) even if we supply
any three qubit state as a blank copy. In section III, we extend
the result to n-qubits and identify the class for which LOCC
cloning is possible with the help of a GHZ state. The possibility
of LOCC cloning of a pair of
orthogonal n-qubit GHZ state is also discussed.\\

\section{Cloning of three qubit CAT states}
The full orthogonal canonical set of three qubit CAT states can be
written as (up to a global phase):
\begin{equation}
\label{GHZs} \ket{\Psi_{p,i,j}(\alpha)}_{123} =(\cos \alpha)^{1-p}
(\sin\alpha)^{p} ~\ket{0~i~j} ~+~ (-1)^{p} ~(\cos \alpha)^{p}
(\sin\alpha)^{1-p} ~ \ket{1~\overline{i}~ \overline{j}},
\end{equation}
where $0<\alpha <\frac{\pi}{4}$ and $p, i, j \in \{ 0, 1\}$,
also the bar over a bit value indicates its logical negation. For
$\alpha=\pi/4$, or $\alpha=0$ the state $\ket{\Psi_{p,i,j}(\alpha)}_{123}$ (given
in (\ref{GHZs})) represents a three qubit GHZ, or a product state respectively.
\\

From a recent result \cite{GYC10} on local cloning of bipartite entangled states one can say that the necessary condition for LOCC cloning of a set of pure n-qubit entangled states with the help of a known n-qubit pure entangled
state is that all states belonging to this set must have equal entanglement. This is due to the fact that, if the LOCC cloning of n-partite states is impossible with $n-1$ cooperating parties, it remains impossible with n independent parties. So for LOCC cloning we always consider a set which contains states with equal entanglement.\\

By appropriate basis transformation any pair of the above (CAT)states can be reduced either to the form:\\

\noindent{\textbf{I}}
{\setlength\arraycolsep{0.01em}
\begin{eqnarray}
\label{3ghzI}
\ket{\Psi_{0ij}(\alpha)}_{123} &=& \cos \alpha \ket{0~i~j}
+ \sin \alpha \ket{1~\overline{i}~ \overline{j}}\nonumber\\
\ket{\Psi_{1ij}(\alpha)}_{123} &=& \sin \alpha \ket{0~i~j} - \cos
\alpha \ket{1~\overline{i}~ \overline{j}}
\end{eqnarray}}
where, $0< \alpha < \frac{\pi}{4}$ and
$i,j\in \{ 0, 1\}$.\\

~~~~~~~~~~~~~~~~~~~~~~~~~~~~~~~~~~~~~~~~~~~~~~~~~~~~~~~or to the form:,\\

\noindent{\textbf{II}}
\begin{eqnarray}
\label{3ghzII}
\ket{\Psi_{0ij}(\alpha)}_{123} &=& \cos \alpha \ket{0~i~j}
+ \sin \alpha \ket{1~\overline{i}~ \overline{j}}\nonumber\\
\ket{\Psi_{0kl}(\alpha)}_{123} &=& \cos \alpha \ket{0~k~l} + \sin
\alpha \ket{1~\overline{k}~ \overline{l}}
\end{eqnarray}
where, $0< \alpha < \frac{\pi}{4}$;  $
i,j,k,l\in \{ 0, 1\}$ and $(i,j)\ne(k,l)$. It needs to be very clear that every pair of states can only be reduced in one of the forms not both.\\

\begin{thm} No \emph{set}\footnote{In this article whenever we mention a `\emph{set}', which always means a set having more than one elements, \emph{i.e.}, a non trivial set.} of orthogonal three qubit \emph{CAT states}\footnote{CAT state, or a set of CAT states, henceforth, we will mean states of kind (\ref{GHZs}) with $0<\alpha< \pi/4$.} can be cloned
by LOCC with the help of a three qubit CAT state having the same
amount of entanglement.\end{thm}

\begin{proof} Without loss of generality we take
our blank copy to be of the form $|\Psi_{0ij}(\alpha)\rangle_{123}=
\cos \alpha |0~i~j\rangle+ \sin \alpha \ket{1~\overline{i}~
\overline{j}}$ (where, $0<\alpha< \pi/4$). Now any set (set containing more than one state)
of orthogonal three qubit CAT states contains at least one of the
pair (either pair(I) or pair(II)). A necessary condition for
cloning of a three qubit state under the usual LOCC (where each of the
three qubits of the state is independently manipulated) would be that:
``The states should remain copyable when two of the three qubits are
operated on jointly at one place whereas the third qubit undergoes a separate
local
operation at a different place and there can be classical communication between these two places".\\

Case-I: The first qubit is kept in
lab-A and the remaining two qubits are kept together in a different
lab(lab-B). Under this arrangement, any given pair of kind
(I),
for a proper choice of basis, can be written in the following bipartite form:\\

\begin{eqnarray}
\label{2ghzI}
\ket{\Psi_{0ij}(\alpha)}_{AB} &=& \cos \alpha \ket{0}_A
\ket{0}_B + \sin \alpha \ket{1}_A \ket{1}_B \nonumber\\
\ket{\Psi_{1ij}(\alpha)}_{AB} &=& \sin \alpha \ket{0}_A \ket{0}_B
- \cos \alpha \ket{1}_A \ket{1}_B
\end{eqnarray}

and the corresponding blank state reduces to the form:
$$|\Psi_{0ij}(\alpha)\rangle_{AB}= \cos \alpha
|0\rangle_A|0\rangle_B+ \sin \alpha \ket{1}_A \ket{1}_B$$(The
subscripts A and B indicates the laboratories occupying the
qubits and $0<\alpha< \pi/4$.)\\

To examine the LOCC cloning in this case we follow the same
procedure as in \cite{choudhary}. We assume the existence of a
cloner which, by LOCC between labs A $\&$ B, can clone a
pair $|\Psi_{0ij}(\alpha)\rangle_{AB}$ and
$|\Psi_{1ij}(\alpha)\rangle_{AB}$ given a known three qubit CAT
state ($|\Psi_{0ij}(\alpha)\rangle_{AB}$) with the same amount of
entanglement supplied to it as blank copy. If we now supply to the
cloner an equal mixture of $|\Psi_{0ij}(\alpha)\rangle_{AB}$ and
$|\Psi_{1ij}(\alpha)\rangle_{AB}$ together with the blank copy
$|\Psi_{0ij}(\alpha)\rangle_{AB}$, \emph{i.e.} if the input state to this LOCC-cloner is:\\
$$\rho_{in}=\frac{1}{2}P[|\Psi_{0ij}(\alpha)\rangle_{AB}\otimes |\Psi_{0ij}(\alpha)\rangle_{AB}]~+
~\frac{1}{2}P[|\Psi_{1ij}(\alpha)\rangle_{AB}\otimes |\Psi_{0ij}(\alpha)\rangle_{AB}],$$\\
then the output of the cloner will be
$$\rho_{out}=\frac{1}{2}P[|\Psi_{0ij}(\alpha)\rangle_{AB}\otimes
|\Psi_{0ij}(\alpha)\rangle_{AB}]~+~\frac{1}{2}P[|\Psi_{1ij}(\alpha)\rangle_{AB}\otimes
|\Psi_{1ij}(\alpha)\rangle_{AB}]$$
 Here $P[~.~]$ stands for projector.\\

To show impossibility of  LOCC-cloning of these states, we make
use of the fact that Negativity of a bipartite quantum state
$\rho$, namely $\emph{N}(\rho)$, cannot increase under LOCC
\cite{vidal}. $\emph{N}(\rho)$ is
given by \cite{zycz}\\
\begin{equation}
\label{nega}
 \emph{N}(\rho)\equiv \|\rho^{T_{B}} \|-1
\end{equation}
 where $\rho^{T_{B}}$ is the partial transpose with respect to
system B and $\|...\|$ denotes the trace norm which is defined as,
\begin{equation}
\label{partial}
\|\rho^{T_{B}}\|=tr(\sqrt{\rho^{T_{B}^{\dagger}}\rho^{T_{B}}}~)
\end{equation}

The negativity of the input state $\rho_{in}$ is
$$\emph{N}(\rho_{in})=\sin{2\alpha}$$
whereas the negativity of the output is
$$\emph{N}(\rho_{out})=\sin{2\alpha}~(\sin{2\alpha}+
2\cos{2\alpha})\hspace{3.2cm}.$$

One can easily check that $\emph{N}(\rho_{in})<
\emph{N}(\rho_{out})$ as,

\begin{eqnarray*}
1&<& \sin{2\alpha}+ 2\cos{2\alpha} ~~for~all ~~0<\alpha<\pi/4 \\
\sin{2\alpha}&<& \sin{2\alpha}~(\sin{2\alpha}+ 2\cos{2\alpha}) .
\end{eqnarray*}

Hence the above cloning is impossible.

\begin{figure}
\includegraphics{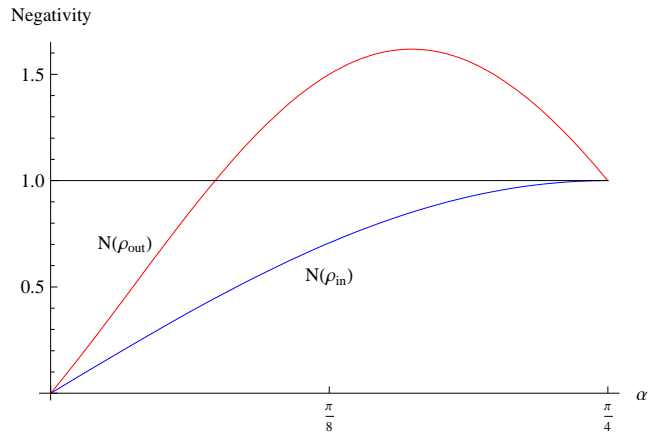}
\caption{\label{fig:negativity2} (Color online) Red line: Plot of
negativity of the output state ($N(\rho_{out})$); Blue line: Plot
of negativity of input state ($N(\rho_{in})$). Negativity of
the output is more than that of the input except for the GHZ case (\emph{i.e.}
$\alpha=\frac{\pi}{4}$) in a given bipartite cut. For a large
range of $\alpha$ the negativity of the output state is more than 1.}
\end{figure}

For a large range of $\alpha$
($\frac{1}{2}\sin^{-1}{(\frac{1}{\sqrt{5}})}<\alpha<\frac{\pi}{4}$),
the negativity of the output state
$N(\rho_{out})=\sin{2\alpha}~(\sin{2\alpha}+ 2\cos{2\alpha})
>1$ (see Fig.
\ref{fig:negativity2}). This implies that even a three qubit GHZ
state (the negativity of a three qubit GHZ state is 1 in any
bipartition) cannot help as a blank copy for those values of
$\alpha$, if the set of states contain the pair (I). In the next
theorem we prove that LOCC cloning of a set of three
qubit CAT states for $0<\alpha<\frac{\pi}{4}$, and with the help of any three qubit
state as a blank
copy, is impossible if the set contains a pair of type (I).\\

Case-II: If the set contains a pair of type (II), then we can
always
find a bi-partition for a proper choice of basis such that the pair of type (II) is reduced to the form:\\

\begin{eqnarray}
\label{2ghzII}\ket{\Psi_{0ij}(\alpha)}_{AB} &=& \cos \alpha \ket{0}_A \ket{0}_B
+ \sin \alpha \ket{1}_A \ket{1}_B
\nonumber\\
\ket{\Psi_{0kl}(\alpha)}_{AB} &=& \cos \alpha \ket{0}_A \ket{1}_B
+ \sin \alpha \ket{1}_A \ket{0}_B \end{eqnarray}

and the corresponding blank state reduces to the form:\\
$$|\Psi_{0ij}(\alpha)\rangle_{AB}= \cos \alpha
|0\rangle_A|0\rangle_B+ \sin \alpha \ket{1}_A \ket{1}_B,$$
where, $0<\alpha< \pi/4$.\\

Let us assume that our cloning machine can clone
$|\Psi_{0ij}(\alpha)\rangle_{AB}$ and
$\ket{\Psi_{0kl}(\alpha)}_{AB}$ if a pure three qubit CAT state
$|\Psi_{0ij}(\alpha)\rangle_{AB}$ is used as a blank copy. So if
the density matrix supplied to this machine is

$$\rho_{in}=\frac{1}{2}P[|\Psi_{0ij}(\alpha)\rangle_{AB}\otimes |\Psi_{0ij}(\alpha)\rangle_{AB}]~+
~\frac{1}{2}P[\ket{\Psi_{0kl}(\alpha)}_{AB}\otimes |\Psi_{0ij}(\alpha)\rangle_{AB}],$$\\
then the output of the cloning machine is
$$\rho_{out}=\frac{1}{2}P[|\Psi_{0ij}(\alpha)\rangle_{AB}\otimes
|\Psi_{0ij}(\alpha)\rangle_{AB}]~+~\frac{1}{2}P[\ket{\Psi_{0kl}(\alpha)}_{AB}\otimes
\ket{\Psi_{0kl}(\alpha)}_{AB}]$$

 Putting $|\Psi_{0ij}(\alpha)\rangle_{AB}$ and $\ket{\Psi_{0kl}(\alpha)}_{AB}$
 in the expression for $\rho_{in}$ and
 $\rho_{out}$ and making use of Eqs. (\ref{nega}) and (\ref{partial}), we get

$$\emph{N}(\rho_{in})=\sin{2\alpha}$$

$$\emph{N}(\rho_{out})=\sqrt{\sin^2{2\alpha}~ (2-\sin^2{2\alpha})}$$

It is easy to check that $\emph{N}(\rho_{in})<
\emph{N}(\rho_{out})$ as,

\begin{eqnarray*}
\sin^4 2\alpha&<& \sin^2 2\alpha~~for~all ~~0<\alpha<\pi/4 \\
\sin^2 2\alpha&<& \sin^2 2\alpha ~(2-\sin^2 2\alpha) .
\end{eqnarray*}
Hence, the above cloning is impossible as it implies
$\emph{N}(\rho_{out})>
\emph{N}(\rho_{in})$ (see Fig. \ref{fig:negativity1}).\\

\begin{figure}
\includegraphics{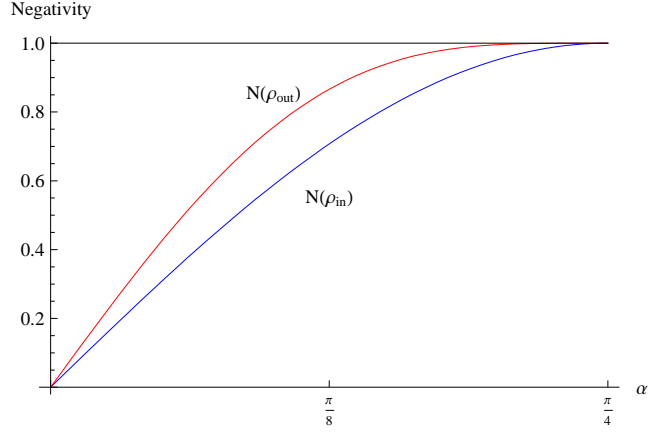}
\caption{\label{fig:negativity1} (Color online) Red line: Plot of
negativity of the output state ($N(\rho_{out})$). Blue line: Plot
of negativity of the input state ($N(\rho_{in})$). Negativity of
output is more than that of the input except for the GHZ case (\emph{i.e.}
$\alpha=\frac{\pi}{4}$) in a given bipartition.}
\end{figure}
\end{proof}

 In the next section we see that LOCC cloning of pair (II) is
possible if we use a GHZ state as blank.\\

\begin{thm} No set of orthogonal three qubit CAT states
can be cloned by LOCC with the help of any three qubit
state, if the set contains the pair (\textbf{I}). \end{thm}

\begin{proof} One needs an entangled blank state to clone an entangled
state otherwise, entanglement of the entire system will increase
under LOCC, which is impossible. We know that any genuine
tripartite entangled state can have entanglement either of the
GHZ-kind or W-kind \cite{dur}. Hence our blank copy can either
belong to the GHZ-class or belong to the W-class. Proof of this theorem
follows immediately from lemmas (2.1), (2.2), and
(2.3), stated below.\\

\begin{lem}\label{2.1} It is impossible to convert the bipartite
state $\ket{\psi}(=\cos{\alpha}\ket{00}+\sin{\alpha}\ket{11})$ to
the state $\ket{\phi}(=\sin{\alpha}\ket{00}-\cos{\alpha}\ket{11})$
by applying an unitary operator on one subsystem only where $0<\alpha<\frac{\pi}{4}$.\end{lem}

\begin{proof} Without loss of generality we choose a general unitary
operator acting only on the first subsystem. Then the unitary
acting on the whole system will be $U \otimes I$, where $I$ is the
$2\times 2$ identity operator acting on the second subsystem. Let
\begin{equation}\label{unonmax}
U \otimes I \ket{\psi} = \ket{\phi}
\end{equation}

for some $2\times 2$ unitary operator $U$ acting on the first
system. The general form of a $2\times 2$ unitary matrix is $U=
\left(\begin{array}{cc}
        a & \lambda b \\
        -b^*& \lambda a^*
        \end{array}
  \right) $, where $a, b , \lambda $ are complex and
  $|a|^2+|b|^2=1=|\lambda|$.\\
  If $\ket{\psi}=(U\otimes I) \ket{\phi}$ holds, then using
  simple algebra \cite{ramij} one can see that (\ref{unonmax})
  holds only when $\alpha=\pi/4$, meaning that the
above mentioned local convertibility is possible only
when both $\ket{\psi}$ and $\ket{\phi}$ are maximally entangled.\end{proof}

\begin{lem} \label{2.2} LOCC cloning of any bipartite set of
orthogonal states is impossible if the set contains both
$\ket{\psi}$ and $\ket{\phi}$ with $0<\alpha<\frac{\pi}{4}$ where\\
\begin{eqnarray}
 \label{nmaxpair}
\ket{\psi} &=&\cos{\alpha}\ket{00}+\sin{\alpha}\ket{11}\nonumber\\
\ket{\phi}&=&\sin{\alpha}\ket{00}-\cos{\alpha}\ket{11}
\end{eqnarray}\end{lem}

\begin{proof} If the above cloning is impossible with a maximally
entangled state as blank copy, then it is also impossible with the
help of any non-maximally blank state, as a non-maximally entangled
state can be produced from a maximal one by LOCC. It has been shown
recently in \cite{kay,GYC10} that the necessary
condition for LOCC cloning of a set of orthogonal bipartite
non-maximally pure entangled states $\{\ket{\psi_j}\}_{j=0}^n$
with the help of a maximally entangled state is
$$\ket{\psi_j}=(U_j\otimes I)\ket{\psi_0}$$
where the $U_j$ are unitary operators acting on the first system of
the bipartite system. Lemma \ref{2.1} says that the pair
(\ref{nmaxpair}) ($\{\ket{\phi}, \ket{\psi}\}$) can not be
expressed as $\ket{\psi}=(U\otimes I)\ket{\phi}$. Therefore, LOCC
cloning of $\ket{\psi}$ and $\ket{\phi}$ is impossible.\end{proof}

\begin{lem} \label{2.3} Any three qubit state can have at most one
ebit of entanglement in any bipartite cut.\end{lem}

For the type (I) pair given in (\ref{3ghzI}), we can always find a
bi-partition for a proper choice of basis such that the pair of
type (I) will be reduced to the form (\ref{2ghzI}), which is
equivalent to pair (\ref{nmaxpair}) ($\{\ket{\phi},\ket{\psi}\}$).
Lemma \ref{2.2} says that LOCC cloning of the pair (\ref{nmaxpair})
($\{\ket{\phi},\ket{\psi}\}$) is impossible with the help of a
two-qubit entangled state as blank (one ebit or less).
Lemma \ref{2.3} demands that any three qubit state has at most one ebit
of entanglement in any bipartite cut. Therefore LOCC cloning of the
pair (I) given in (\ref{3ghzI}) is impossible with the help of any
three qubit state.\end{proof}

LOCC cloning of 3-qubit orthogonal CAT states becomes trivial if there is no restriction on the supply of
free entanglement, e.g., if parties `1 \& 2' and parties `2 \& 3' shared a sufficient number of Bell states then teleportation can be used to bring all the subsystem
to the second party's lab and then to send back after cloning. If parties `1 \& 2' (or, `2 \& 3') shared just one copy of Bell state then LOCC cloning of the pair (\ref{3ghzI}) remains impossible as the supply of free entanglement in 1 vs. 23(or, 3 vs. 12) bipartite cut is one ebit which is not sufficient for pair (\ref{3ghzI}).

\section{Cloning of $n$-qubit CAT states}
The full orthogonal canonical set of $n(\ge 2)$-qubit CAT states can be
written as (up to a global phase):
\small{
\begin{equation}
\label{nGHZs}
\ket{\Psi_{i_1,i_2,i_3,.....,i_n}(\alpha)}_{123....n} =(\cos
\alpha)^{1-i_1} (\sin\alpha)^{i_1} ~\ket{0~i_2~i_3~......i_n} ~+~
(-1)^{i_1} ~(\cos \alpha)^{i_1} (\sin\alpha)^{1-i_1} ~
\ket{1~\overline{i_2}~ \overline{i_3}.......\overline{i_3}},
\end{equation}}
\normalsize
where $0< \alpha < \frac{\pi}{4}$ and $i_1, i_2, i_3,.....i_n \in \{ 0, 1\}$, and
the bar over a bit value indicates its logical negation.\\

By appropriate basis transformation any pair of the above (CAT)states can be reduced either to the form:\\

\noindent{\textbf{I}}
\begin{eqnarray}
\label{nghzI}\ket{\Psi_{0i_2 i_3....i_n}(\alpha)}_{123.....n} &=& \cos \alpha \ket{0~i_2~ i_3....i_n}
+ \sin \alpha \ket{1~\overline{i_2}~
\overline{i_3}.....\overline{i_n}}\nonumber\\
\ket{\Psi_{1i_2 i_3....i_n}(\alpha)}_{123.....n}&=& \sin \alpha
\ket{0~i_2~i_3.....i_n} - \cos \alpha \ket{1~\overline{i_2}~
\overline{i_3}.....\overline{i_n}}
\end{eqnarray}
where, $0< \alpha < \frac{\pi}{4}$ and
$i_2, i_3,....,i_n\in \{ 0, 1\}$,\\

~~~~~~~~~~~~~~~~~~~~~~~~~~~~~~~~~or to the form:\\

\noindent{\textbf{II}}
\begin{eqnarray}
\label{nghzII}\ket{\Psi_{0i_2i_3.....i_n}(\alpha)}_{123....n} &=& \cos \alpha
\ket{0~i_2~i_3.....i_n} + \sin \alpha \ket{1~\overline{i_2}~
\overline{i_3}.....\overline{i_n}}\nonumber\\
\ket{\Psi_{0j_2j_3....j_n}(\alpha)}_{123....n} &=& \cos \alpha
\ket{0~j_2~j_3.....j_n} + \sin \alpha \ket{1~\overline{j_2}~
\overline{j_3}.....\overline{j_n}} \end{eqnarray}
where, $0< \alpha < \frac{\pi}{4}$;  $
i_2, i_3,.....i_n,j_2,j_3,.....,j_n \in \{ 0, 1\}$ and
$(i_2,i_3,~....~,i_n)\ne(j_2,j_3,~....~,j_n)$.\\

\begin{thm} No set of n($\ge 2$)-qubit CAT-states can be cloned by LOCC
with the help of a $n$-qubit GHZ-state if the set
contains the pair(\textbf{I}) (given in
(\ref{nghzI})).\end{thm}

\begin{proof} Cloning of these states by LOCC implies cloning of them in
any bipartition. Keeping this in mind we put the first qubit in
lab-A and the remaining $n-1$ qubit in lab-B. Under this
arrangement, any given pair of type (I) (given in
(\ref{nghzI})), for
a proper choice of basis, reduces to the form:\\

\begin{eqnarray}
\label{n2ghzII}
\ket{\Psi_{0i_2i_3....i_n}(\alpha)}_{123....n} &=& \cos \alpha
\ket{0}_A \ket{0}_B + \sin \alpha \ket{1}_A \ket{1}_B \nonumber\\
\ket{\Psi_{1i_2i_3....i_n}(\alpha)}_{123....n} &=& \sin \alpha
\ket{0}_A \ket{0}_B - \cos \alpha \ket{1}_A
\ket{1}_B,\end{eqnarray}
 which is equivalent to pair (\ref{nmaxpair}). Lemma \ref{2.3} also holds for any $n$-qubit GHZ-state (\emph{i.e.}, n-qubit GHZ
 state has a maximum of one ebit of entanglement in any bipartite cut). Therefore, using Lemma \ref{2.2} we conclude
 that LOCC cloning of pair (I) given in (\ref{nghzI}), and with the help of any $n$-qubit GHZ state as blank, is impossible. \end{proof}

\begin{thm} Any set ($S(\alpha)$, with $0< \alpha < \pi/4$) of orthogonal n($\ge 2$)-qubit
CAT states (having equal entanglement) can be cloned by LOCC with
the help of a n-qubit GHZ-state if the set ($S(\alpha)$)
does not contain a pair of kind (\textbf{I}), where $\alpha$ is
the entanglement parameter
. \end{thm}

\begin{proof} To prove this theorem we demonstrate a protocol where
such a LOCC cloning is possible with the help of a $n$-qubit
GHZ-state. The protocol is similar to the protocol given by Kay
and Ericsson \cite{kay}, which was used for LOCC cloning of a set
of non-maximal bipartite states. It is easy to see that if the set
($S(\alpha)$, with $0<\alpha<\frac{\pi}{4}$) of orthogonal $n$-qubit CAT-states (having equal
entanglement) does not contain a pair of kind (I) then any state
$\ket{\Psi_{i_0i_2i_3....i_n}(\alpha)}_{123....n} \in S(\alpha)$
obeys the relation $$
\ket{\Psi_{0i_2i_3....i_n}(\alpha)}_{123....n} = (I\otimes
P_{i_2}\otimes P_{i_3}.....\otimes P_{i_n})\ket{\Psi_{000....0}
(\alpha)}_{123....n}$$ $$P_{i_j}=\sum_{k=0}^{1}{|i_j+k ~mod
~2\rangle\langle k|}$$
 where $j=2,3,...n$, and $i_2, i_3,....,i_n\in \{ 0, 1\}$.

The maximum cardinality of such a set($S(\alpha)$) of non-maximally entangled
orthogonal GHZ-states is $2^{n-1}$. So our goal is to perform the
transformation under LOCC such that,
$$\ket{\Psi_{i_0i_2i_3....i_n}
(\alpha)}_{123....n}\ket{\Psi_{000....0}(\pi/4)}_{123....n}
\rightarrow \ket{\Psi_{i_0i_2i_3....i_n}(\alpha)}_{123....n}
\ket{\Psi_{i_0i_2i_3....i_n}(\alpha)}_{123....n},$$ for all
$\ket{\Psi_{i_0i_2i_3....i_n}(\alpha)}_{123....n} \in S(\alpha)$,
where $\ket{\Psi_{000....0}(\pi/4)}_{123....n}$
is a n-qubit GHZ-state, which is our blank state.\\

To fulfill the above task we first perform the transformation
$$\ket{\Psi_{i_0i_2i_3....i_n}
(\alpha)}_{123....n}\ket{\Psi_{000....0}(\pi/4)}_{123....n}\rightarrow
\ket{\Psi_{i_0i_2i_3....i_n}(\alpha)}_{123....n}\ket{\Psi_{i_0i_2i_3....i_n}(\pi/4)}_{123....n}$$
where $\ket{\Psi_{i_0i_2i_3....i_n}(\pi/4)}_{123....n}=(I\otimes
P_{i_2} \otimes P_{i_3}.....\otimes
P_{i_n})\ket{\Psi_{000....0}(\pi/4)}_{123....n}$. This
transformation can be realised if all n parties apply the CNOT
operation on their respective qubit by taking the unknown original
state's qubit as source and the qubit  of the blank (GHZ) state as
target. Now we convert the GHZ state to the corresponding CAT
state. This is achieved by the first party applying the
measurement operators
$$M_0=\cos{\alpha}|0\rangle\langle 0|+
\sin{\alpha}|1\rangle\langle 1|,~ M_1=\sin{\alpha}|0\rangle\langle
0|+ \cos{\alpha}|1\rangle\langle 1|.$$

If the first party gets the result $k(=0,1)$, then all parties
apply the unitary operation $P_{k_j}$ on their respective qubits.
This completes the cloning protocol. \end{proof}
This theorem is also true for $\alpha=\pi/4$, \emph{i.e.}, for set of $n$-qubit GHZ states.\\

Choudhary \emph{et al.} \cite{choudhary1} give a protocol for LOCC cloning of a pair of maximal three qubit GHZ states.
So the important question is whether a similar protocol exists for a given pair of $n$-qubit maximal GHZ states or not.\\

\begin{thm} Any two orthogonal $n(\ge 2)$-qubit GHZ-states
can be cloned by LOCC with the help of a GHZ state.\end{thm}

\begin{proof} With proper choice of basis transformation any pair of $n(\ge 2)$-qubit
GHZ state can be reduced either to the form
\\

\noindent{\textbf{I}}
\begin{eqnarray*}
\ket{\Psi_{0~0~0....0}}_{123....n} &=&\frac{1}{\sqrt{2}} [
\ket{0~0~0....0} + \ket{1~1~1....1}] \\
\ket{\Psi_{1~0~0....0}}_{123....n} &=&\frac{1}{\sqrt{2}} [
\ket{0~0~0....0} - \ket{1~1~1....1}]
\end{eqnarray*}
~~~~~~~~~~~~~~~~~~~~~~~~~~~~~~~~~or to the form:\\

\noindent{\textbf{II}}
\begin{eqnarray*}
\ket{\Psi_{0~0~0....0}}_{123....n} &=&\frac{1}{\sqrt{2}} [
\ket{0~0~0....0} + \ket{1~1~1....1}]\\
\ket{\Psi_{0i_2 i_3....i_n}}_{123....n} &=& \frac{1}{\sqrt{2}}[\ket{0i_2i_3....i_n}
+ \ket{1\overline{i_2}~
\overline{i_3}....\overline{i_n}}]
\end{eqnarray*}
where, $i_2, i_3,.....i_n \in \{ 0, 1\}$ and
$(i_2, i_3,.....,i_n) \neq (0,0,.....,0)$.\\

Without loss of generality, we assume that our blank state will be
$\ket{\Psi_{0~0~0....0}}_{123....n} =\frac{1}{\sqrt{2}} [
\ket{0~0~0....0} + \ket{1~1~1....1}]$ for both cases.\\

For pair (I), each the party applies CNOT operations to their
respective qubits separately by taking the blank state's qubit as
source and the original state's qubit as target to achieve the
required cloning. In contrast, for pair (II), each of them (all the parties) will apply the CNOT operations by taking the unknown original
state's qubit as source and blank state's qubit as target.\end{proof}

Although any two orthogonal $n$-qubit GHZ states are LOCC cloneable, for every $n$, there always exists three orthogonal $n$-qubit GHZ
states that, with the help of any $n$-qubit
GHZ state as blank, cannot be cloned by LOCC. To show this impossibility we choose
$\ket{\Psi_{000.....0}}_{123....n}$,
$\ket{\Psi_{010.....0}}_{123....n}$ and
$\ket{\Psi_{100.....0}}_{123....n}$. These three states are
equivalent to the three Bell states if we choose the bipartition in
such a way that the first qubit is kept in lab-A and the remaining $n-1$
qubits are kept together in lab-B. Now as even 1 free ebit is not
sufficient to clone three Bell states \cite{owari,choudhary} and
as the maximum bipartite entanglement that a $n$-qubit GHZ state can
have is one ebit, so we conclude that these $n$-qubit GHZ states,
with the help of any $n$-qubit GHZ ancilla  cannot be cloned by
LOCC.

\section{Discussion}

In this paper, we have discussed the problem of LOCC cloning of
CAT states. No set of orthogonal three qubit CAT states can be cloned under LOCC with the help of a known CAT state having the
same amount entanglement. The cloning remains impossible for the set of states having a pair of kind (\ref{3ghzI}) even if we supply any three qubit entangled state as
a blank copy. But LOCC cloning is possible for a particular type
of set (set does not contain a pair of kind (\ref{nghzI})) of orthogonal CAT states (having the same amount of
entanglement) with the help of a GHZ state. Any pair of n-qubit
CAT states of type (\ref{nghzI}) cannot be cloned with the help
of a n-qubit GHZ state
whereas any two n-qubit GHZ states can be cloned by LOCC with the aid of a known GHZ state.\\

\begin{acknowledgments}
The research has been supported by Norwegian Research Council,
Norway. The author thanks G. Kar for several discussions during
his visit at Indian Statistical Institute, Kolkata. The author also thanks M. G. Parker for stimulating discussions.
\end{acknowledgments}

\end{document}